\documentclass[12pt, draftclsnofoot, onecolumn]{IEEEtran}
 
\usepackage{cite}
\ifCLASSINFOpdf
\else
\fi
\usepackage{caption}
\usepackage{graphicx}
\usepackage[cmex10]{amsmath}
\usepackage{mdwmath}
\usepackage{mdwtab}
\usepackage{bm}
\usepackage{amsmath}
\usepackage{amssymb}
\usepackage{multicol} 
\usepackage{float}
\usepackage{epstopdf}
\usepackage{cancel}
\usepackage{cases} 

\newenvironment{proof}[1][Proof]{\textbf{#1.} }{\ \rule{0.5em}{0.5em}}  
\newtheorem{Proposition}{Proposition}

\newcommand{\snr}{\mathrm{SNR}}
 
\newcommand{\data}{\mathrm{DATA}}
\newcommand{\hd}{\mathrm{HD}}
\newcommand{\cab}{\mathrm{CAB}}
\newcommand{\zf}{\mathrm{ZF}}
\newcommand{\ini}{\mathrm{INI}}
\newcommand{\tr}{\mathrm{tr}}

\begin{document} 
\title{Multiuser MIMO Beamforming with Full-duplex Open-loop Training}

\author{Xu Du, John Tadrous, Chris Dick, and Ashutosh Sabharwal 
\thanks{Xu Du, John Tadrous, Ashutosh Sabharwal are with the Department of Electrical and Computer Engineering, Rice University, Houston, TX, 77005. E-mails: \{xu.du, jtadrous, ashu\}@rice.edu. This work was partially supported by a grant from Xilinx Incorporated, and NSF Grant CNS-1314822}
\thanks{Chris Dick is with Xilinx, Inc., San Jose, CA, 95124 USA. E-mail: chris.dick@xilinx.com}
\thanks{The material in this paper was presented (without proof) in 16th IEEE International Workshop on Signal Processing Advances in Wireless Communications, 2015.}
}

\bibliographystyle{IEEEtran}

\maketitle 
\begin{abstract}
In this paper, full-duplex radios are used to continuously update the channel state information at the transmitter, which is required to compute the downlink precoding matrix in MIMO broadcast channels. The full-duplex operation allows leveraging channel reciprocity for open-loop uplink training to estimate the downlink channels. However, the uplink transmission of training creates interference at the downlink receiving mobile nodes. We characterize the optimal training resource allocation and its associated spectral efficiency, in the proposed open-loop training based full-duplex system. We also evaluate the performance of the half-duplex counterpart to derive the relative gains of full-duplex training. Despite the existence of the inter-node interference due to full-duplex, significant spectral efficiency improvement is attained over half-duplex operation. 
\end{abstract}
 
\IEEEpeerreviewmaketitle
\newtheorem{Definition}{Definition}
\newtheorem{theorem}{Theorem}
\newtheorem{Lemma}{Lemma}
\section{Introduction}
In multi-user MIMO, up to $M$ users can be simultaneously supported at full multiplexing gain by an $M$-antenna base station, if perfect channel knowledge is available at such base station. Accurate channel state information (CSI) is essential to achieve maximal multiplexing gain performance. As a result, larger number of antennas demand more CSI, that in turn implies longer feedback duration~\cite{caire2010multiuser}. In a time-division duplex uplink/downlink transmission, longer feedback duration implies reduced time for sending actual data, creating a tradeoff between available amount of CSI and resulting spectral efficiency. In this paper, we propose the use of full-duplex transmission capabilities for simultaneous feedback and data transmission while utilizing the channel reciprocity offered by the same-band operation of full-duplex.


The use of full-duplex radios for multiuser MIMO to reduce the time cost of digital feedback  was first proposed in \cite{du2014mimo}. Even though the full-duplex feedback introduces inter-node interference (INI), it was shown that by refining precoding matrix continuously during digital feedback, full-duplex radios provide a multiplexing gain over their equivalent half-duplex counterparts. In this paper, we extend the continuously adaptive beamforming (CAB) strategy, proposed in~\cite{du2014mimo}, to a multiuser MIMO downlink that exploits channel reciprocity by adopting open-loop training. Open-loop training stands for a system whereby the base station learns CSI by estimating training pilots sent over the uplink channel and then uses the channel estimates for downlink transmissions. Our main contributions are as follows.
 \begin{enumerate} 
\item We extend the CAB strategy of \cite{du2014mimo} to an open-loop training where users exploit the uplink/downlink reciprocity, offered by full duplex, to train the base station of the downlink channel. Unlike quantize and feedback scheme used in \cite{du2014mimo}, open-loop training provides a simpler and faster means to channel state feedback.  We show that open-loop training with CAB has a considerable potential to boost downlink data rates.
\item We study the optimal training duration together with associated spectral efficiency gains of the CAB strategy. We analytically characterize tight approximations of the optimal feedback resources and establish a tight upper bound for the spectral efficiency loss when compared to the genie aided full CSI scenario.
\item We derive a lower bound for the spectral efficiency gains reaped from full-duplex feedback when compared to the half-duplex counterpart. These gains are significant in low-to-moderate signal-to-noise ratio (SNR) regimes, and grow linearly with the number of training symbols.
\end{enumerate} 

In \cite{caire2010multiuser}, different feedback training strategies for half-duplex MIMO broadcast channels have been characterized, though the focus was not on optimizing feedback resources. In~\cite{kobayashi2011training}, the authors considered the optimization problem of time resource allocation for systems with different forms of feedback. The major difference between \cite{kobayashi2011training} and this paper is that we consider full-duplex systems instead of half-duplex systems. In addition, we do not require that users send feedback with equal power to that of the base station. Instead, we analyze our system under a general choice of feedback power which is allowed to be smaller than the base station's.

The rest of this paper is organized as follows. In Section~\ref{sec:model}, we introduce the setup of a full-duplex multiuser MIMO downlink system. In Section~\ref{sec:Result}, we first extend the CAB strategy to a system with open-loop training where channel reciprocity holds. Then the performance of CAB is evaluated both analytically and via simulations. The paper is concluded in Section~\ref{sec:conclude}.

\section{System Model} \label{sec:model}
Consider a full-duplex multiuser MIMO downlink system consisting of an $M$-antenna base station, streaming downlink data to $M$ single-antenna users. The base station relies on the downlink channel knowledge to construct the precoding matrix. Albeit sub-optimal, we consider the zero-forcing (ZF) precoding strategy \cite{spencer2004zero} for its simplicity. Imperfect CSI yields inaccurate ZF beamforming, which then translates to inter-beam interference (IBI).

Since full-duplex enables simultaneous transmission on both uplink and downlink over the same band, two new types of interference may result. The first one is self-interference, which is the interference caused by the transmitter to its own receiver. We assume that self-interference can be reduced to near-noise floor by a combination of active cancellation \cite{Duarte2010FD, Jain2011practicalFD} and passive suppression \cite{everett2014passive}. The other is inter-node interference (INI), caused by uplink training pilots sent concurrently during downlink transmissions.  While sending more feedback is advantageous to reduce IBI, it also induces INI during training. In this paper, we focus on the tradeoff between IBI and INI, and assume perfect self-interference cancellation at all nodes\footnote{We further discuss performance of systems consisted of half-duplex nodes, where no user self-interference cancellation is needed, in~\cite{FDMIMO2015trans}.} in the system. Fig.~\ref{fig:interference} gives a high level schematic of different types of interference.

\begin{figure}[htbp]
\centering
\includegraphics[width=8cm, ,height=4cm]{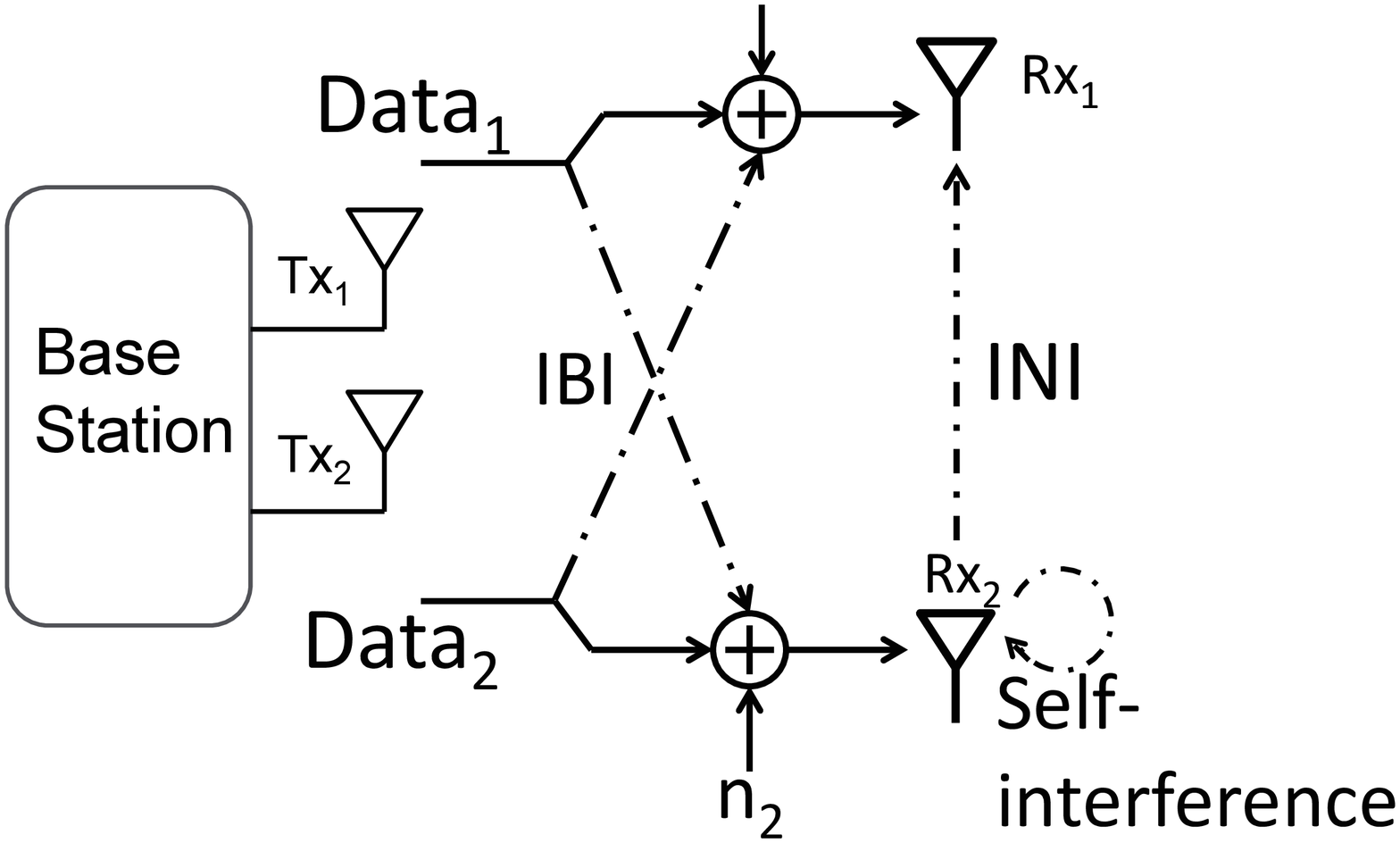} 
\caption{A schematic of different types of interference in a $2 \times 2$ full-duplex MIMO broadcast system. } 
\label{fig:interference}
\end{figure}

We assume users send training symbols sequentially in orthogonal time slots. Then the signal received by user $i$, when user $k$ is transmitting uplink training symbols, is given by
\begin{equation}
y_{i}=\left\{\begin{matrix}
\textbf{h}_{i}^{*}\textbf{x}+h_{ik}x_{\mathrm{UP}_{k}}+n_{i},\ & k\not=i \\ 
\textbf{h}_{i}^{*}\textbf{x}+n_{i},\ &k=i
\end{matrix}\right. ,\quad i=1,\cdots,M,
\end{equation}
where $\textbf{h}_{i}\in \mathbb{C}^{1\times M}$ and $\textbf{x}$ denotes the channel realization of user $i$ and the precoded signal matrix, respectively. Thus the term $\textbf{h}_{i}^{*}\textbf{x}$ includes both the intended signal and the IBI due to imperfect channel knowledge. The signal  $x_{\mathrm{UP}_{k}}$ is the uplink training sent by the $k$-th user. The channel between the $k$-th user and $i$-th user is denoted by $h_{ik}$. Finally, the signal is degraded by an independent unit variance additive white complex Gaussian noise samples $n_{i}$. The fading environment is assumed to be Rayleigh block fading, i.e., each element of $\textbf{h}_{i}$ is independently complex Gaussian distributed from block to block. The total block length is $T$ symbols.

The transmit signal from the base station is subject to an average power constraint ($\mathbb{E}[|\textbf{x}|^2]\leqslant P$). In this paper, we consider equal power allocation among antennas and symbols. Due to the limitations of both battery and size of user devices, a more strict average power constraint is considered at the users, which is mathematically captured as $\mathbb{E}[|x_{\mathrm{UP}_{k}}|^2]\leqslant fP,\ f\in (0,1]$. 
The INI incurred by uplink training pilots is assumed proportional to the training power, and grows as $\alpha f P$, where $\alpha>0$. Finally, we assume that each user has perfect knowledge of its own channel realization. 

\section{Main Results} \label{sec:Result}
In this section, we extend the CAB strategy in \cite{du2014mimo} to systems with open-loop training, and evaluate their performance. Since there is INI generated by full-duplex operation, the effect of it is first quantified in Section~\ref{sec:INI}. In Sections~\ref{sec:OptLen} and~\ref{sec:SpEf}, we focus on the time resource allocation and spectral efficiency of the CAB strategy. The result for the half-duplex counterpart is also presented and compared with CAB.

\subsection{Continuously Adaptive Beamforming} \label{sec:CAB}

We first describe Continuous Adaptive Beamforming (CAB) strategy for open-loop training system. The key idea behind the CAB is to send downlink data concurrently with uplink training pilots collection, instead of waiting for all the uplink training pilots to be collected and then starting downlink transmission. For a transmission block with $T_{\cab}$ uplink training pilots, CAB with open-loop training operates as follows.

\begin{enumerate}

\item At the beginning of each block,  no CSI is available at the base station. Each of users $1$, ..., $M$  sends a training symbol sequentially in a TDMA manner from symbol $1$ to symbol $M$. We refer to such $M$ symbols where each user sends one training pilot as a \textsl{cycle}. The first $M$ symbols constitute cycle $j=1$.

\item For cycle $j+1$, the base station updates its ZF precoding matrix and transmits downlink data based on all the uplink training pilots received during the whole previous $j$ cycles, i.e., $j$ uplink training pilots from each user\footnote{It is also possible to update precoding matrix at the end of each symbol. Albeit more promising than updating at the end of each cycle, the spectral efficiency performance is not significantly different. We will discuss such system in our extended work.}. All users decode the received signal by treating interference as noise. An illustration is shown in Fig.~\ref{fig:CAB}.

\item Repeat 2) till the end of $T_{\cab}$ symbols. 

\item When all $T_{\cab}$ uplink training pilots are collected, all users stop sending training pilots to the base station and transmission continue in the downlink direction only.

\end{enumerate}

\begin{figure}[htbp]
\centering
\includegraphics[width=8cm,height=3cm]{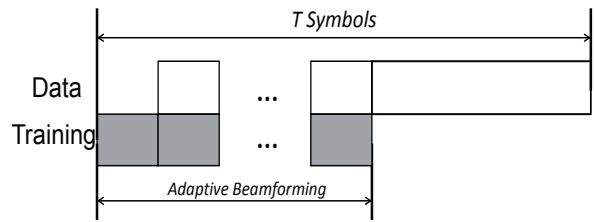}
\caption{A $T_{\cab}$-symbol ($T_{\cab}/M$-cycle) CAB is illustrated. At the end of cycles, the base station refines its precoding vectors based on the accumulated pilots.}
\label{fig:CAB}
\end{figure} 
In each cycle, there is one symbol where the user itself sends uplink training pilot to the base station. We note the rate performance in this symbol as $R_{\cancel{\ini}}$, since no INI is incurred from other users. However, since no user cooperation is applied, in the rest $M-1$ symbols, the user suffers INI induced by training symbols sent by other users. The associated rate in those $M-1$ symbols is referred as $R_{\ini}$. The spectral efficiency (achievable rate) attained by adopting the CAB strategy can be then characterized as\footnote{Traditionally, the ergodic rate is shown to be achievable by coding over blocks with a simple Gaussian long codebook\cite{ErodAchGauLon1999Shamai}. Similarly, CAB also relies on coding downlink symbols across different blocks.}:
\begin{multline} 
AR_{\cab}=\frac{T-T_{\cab}}{T}R_{\data}+\sum_{j=2}^{T_{\cab}/M}R_{\cancel{INI}}((j-1)M)+\left(M-1\right)R_{INI}((j-1)M).
\label{equ:ARCAB}
\end{multline} 
In~\eqref{equ:ARCAB}, the argument $(j-1)M$ captures the number of training symbols received so far by the base station. Also, $R_{\data}$ is the downlink data transmission rate achieved after all $T_{\cab}$ uplink training pilots have been collected.  Assuming MMSE estimation at the base station, $R_{\data}$ is obtained below by following the recipe in \cite{caire2010multiuser}.
\begin{align}
R_{\data}\left(T_{\cab}\right)=&R^{\zf}-\Delta R_{\data}\left(T_{\cab}\right) \notag \\
\geqslant& R^{\zf}- \log\left(  1+\frac{P}{M}\frac{M-1}{1+T_{\cab} f P/M}\right).  \label{equ:dRDATA}
\end{align}
The term $R^{\zf}$ stands for the rate when a genie provides perfect CSI to the base station. The term $\Delta R_{\data}\left(T_{\cab}\right)$ is the rate loss due to imperfect channel knowledge. Due to the limit of block length and user power, finite training duration and feedback power lead to a positive rate loss.

The spectral efficiency of the half-duplex counterpart is given by
\begin{equation} 
AR_{\hd}=\frac{T-T_{\hd}}{T}R_{\data}.
\label{equ:ARHD}
 \end{equation} 
Here $T_{\hd}$ is the duration of training in half-duplex systems. The term $\frac{T-T_{\hd}}{T}$ is the fraction of time  devoted to downlink data transmission out of each coherence time. Hence half-duplex downlink data rate \emph{after} the collection of all training pilots satisfies~\eqref{equ:dRDATA} with $T_{\cab}$ replaced by $T_{\hd}$. 

\subsection{Impact of INI} \label{sec:INI}
Full-duplex radios enable the base station to send downlink data during training. During the symbols when user $i$ sends uplink training pilots, the rate $R_{\cancel{INI}}$ can be readily captured by~\eqref{equ:dRDATA}; note that we assume perfect self-interference cancellation. However, the received signal suffers interference from the uplink training pilots transmitted from other users. Since only imperfect CSI is available, the rate $R_{\ini}$ is the result of both INI and IBI. 
In this paper, we consider INI to be a linear fraction of the total transmit power of each user. In particular, we capture INI as $\alpha f P$, $\alpha>0$, where factor $\alpha$ represents the strength of incurred interference. Next, we provide a theorem that captures the impact of INI in the full-duplex case.

\begin{theorem} \label{thm:INI}
If base station has $\beta$ (integer) uplink training pilots of power $fP$ from each user, the rate per user with INI, $R_{\ini}(\beta M)=R^{\zf}-\Delta R_{\ini} \left(\beta M\right)$, satisfies:
\begin{equation}
R_{\ini} (\beta M) \geqslant R^{\zf}-  \log\left( \frac{1+\frac{P}{M}\frac{M-1}{1+\beta f P} +\alpha f P}{1 +\frac{\alpha f P}{1+\frac{P}{M}}}\right)  \label{equ:dRBT},
\end{equation}
where $\Delta R_{\ini}(\beta M)$ represents the rate loss due to both IBI and INI during full-duplex open-loop training with $\beta M$ training symbols received by the base station.
\end{theorem}
\begin{proof}
By evaluating the signal received at each time, the rate can be then captured with the expectation of interference power by using Jensen's inequality. The interference of IBI follows from~\cite{caire2010multiuser} and INI power is immediate from our assumption. Detailed proof can be found in Appendix~\ref{app:INI}.
\end{proof}

 The impact of IBI and INI appears in terms $\frac{P}{M}\frac{M-1}{1+\beta f P}$ and $\alpha f P$ respectively. The first term decreases with $\beta$, which means that more uplink training pilots enhance CSI estimation accuracy and diminish IBI. The $\alpha f P$ term on the numerator captures the INI impact. While it increases with $\alpha$, another $\alpha$ factor in the denominator poses a finite bound on the rate-loss due to increasing INI. For example, as $\alpha \to \infty$, the rate loss term is upper bounded by $\log\left(1+P/M\right)$, which is obviously finite.

An interesting observation is that as $\beta \to \infty$, the effect of IBI vanishes, and the upper bound on the rate loss becomes 
\begin{equation}
\label{eq:tilde_delta_R}
\widetilde{\Delta R}_{\ini}=\log\left( \frac{1+\alpha f P}{1 +\frac{\alpha f P}{1+\frac{P}{M}}}\right).
\end{equation}
$\widetilde{\Delta R}_{\ini}$ can be viewed as the ``invariant'' rate loss caused by INI during training. This quantity will be relevant in our investigations of the optimal tradeoff between IBI and INI as follows. 

\subsection{Optimal Training Resource Allocation} \label{sec:OptLen}
Using the above rate characterization, we next answer the question: how many symbols $T_{\cab}$, $T_{\hd}$ should be used for training to maximize the spectral efficiency, i.e.,
\begin{equation}
\max_{0 \leq T_{s}< T} AR_{s},\quad s\in\{\cab,\hd\} .
\end{equation}
Shorter training results in larger IBI, while longer training suggests strong influence of INI. Thus, solving for the optimal number of training symbols optimizes the tradeoff between IBI and INI. We first focus on the time allocation for the CAB strategy, and then an approximate solution for half-duplex scenario  is developed to compare with full-duplex.
\subsubsection{CAB Strategy} \label{sec:optLenCAB}
We now consider the question of time resource allocation for a CAB strategy described in section~\ref{sec:CAB}.
\begin{theorem}\label{thm:LenCAB}\label{sec:optlen}
The optimal training duration $T_{\cab}^*$ that maximizes spectral efficiency of a CAB system is approximated as
\begin{equation}
\widetilde{T}_{\cab}^{*}=\frac{\sqrt{4cT+1}-1}{2c}\approx \sqrt{\frac{MT}{f \widetilde{\Delta R}_{\ini} }},
\end{equation}
where $c=\frac{f\widetilde{\Delta R}_{\ini}}{M}.$
\end{theorem}
 \begin{proof}  
The benefit (additional spectral efficiency) from an additional training cycle is monotonically decreasing, due to shorter time for data transmission after training and the convexity of $\log\left(1+1/x\right)$. Moreover, since INI term $\alpha f P$, is independent of $T_{\cab}$, the spectral efficiency $AR_{\cab}$ is concave in $T_{\cab}$.  Therefore, there will be a unique $T_{\cab}^*$ that maximizes the spectral efficiency, at which 
\begin{equation}
\frac{\partial AR_{\cab}\left(T_{\cab}^*\right)}{\partial T_{\cab}}=0.
\end{equation}
Since direct application of the differential operation to the summation in $AR_{\cab}$ is challenging, the following approximation is made for mathematical tractability
\begin{equation}\label{equ:optLenCABcom}
AR_{\cab}\left(T_{\cab}^{*}\right)\approx AR_{\cab}\left(T_{\cab}^{*}+M\right).
\end{equation}
 
Expanding both sides of~\eqref{equ:optLenCABcom} leads to and applying rate loss characterization given by~\eqref{equ:dRDATA} and~\eqref{equ:dRBT}, we have
\begin{align}
&\frac{T-\left(\widetilde{T}_{\cab}^{*}+M\right)}{T}\left[\Delta R_{\data}\left(\widetilde{T}_{\cab}^{*}\right)-\Delta R_{\data}\left(\widetilde{T}_{\cab}^{*}+M\right)\right] \notag \\ =& \frac{M-1}{T}\left[\Delta R_{\ini}\left(\widetilde{T}_{\cab}^{*}\right)-\Delta R_{\data}\left(\widetilde{T}_{\cab}^{*}\right)\right]. \label{equ:optCABcom} 
\end{align}
The rate loss difference on the left hand side ($\text{LHS}$) of~\eqref{equ:optCABcom} can be simplified by using Maclaurin expansion as
\begin{align*}
\text{LHS}= 
\approx&  \frac{(M-1)P}{M } \frac{fP}{\left(1+f\widetilde{T}_{\cab}^{*} P/ M\right)^2}.
\end{align*}

The rate loss difference on the right hand side ($\text{RHS}$) of~\eqref{equ:optCABcom} can also be computed as follows,
\begin{align*}
\text{RHS}= \log\left( \frac{1+\frac{\alpha f P}{1+\frac{P}{ M}\frac{M-1}{1+f  \widetilde{T}_{\cab}^{*} P/ M }}}{1 +\frac{\alpha f P}{1+\frac{P}{M}}}\right) \approx  \widetilde{\Delta R}_{\ini}.
\end{align*}
This means that the RHS is mainly contributed by the INI, not IBI.

Substituting the expressions of LHS, and RHS into~\eqref{equ:optCABcom} leads to the theorem.
\end{proof}

 In the proof above, we observe that the left hand side of~\eqref{equ:optCABcom} is the benefit of reducing IBI on spectral efficiency by adding one training cycle. This can be interpreted as \textsl{marginal utility}. On the right hand side is the spectral efficiency loss due to INI in this new cycle, which can be labeled as \textsl{marginal cost}. The optimal point occurs at the spot where marginal cost equals marginal utility, which suggests the rate benefit obtained by reducing IBI and rate loss due to more INI break even.  

Other interesting observations are as follows. i) Even for large $T$, the optimal training duration scales in the order of $\sqrt{T}$. Such scaling has been observed for both half-duplex MIMO broadcast channels with analog feedback\cite{kobayashi2011training} and point-to-point MIMO\cite{hassibi2003much}. It suggests an identical scaling law is shared by half-duplex and full-duplex MIMO systems. A numerical simulation result provided in Fig.~\ref{fig:OptLenT} confirms this analysis. ii) The INI has a crucial impact on the optimal training duration. A large $\alpha$ (high INI) implies high rate loss, which leads to shorter feedback duration. In contrast, small $\alpha$ (low INI) yields longer training duration. When $\alpha=0$ (users are hidden from each other), users should always send the training pilots. iii) The feedback power also plays a role in the optimal training length. Small $f$ leads to low INI power (rate loss), which contributes to a longer training duration to enhance the IBI performance. iv) As $\snr$ grows, the optimal duration of training decreases due to high-quality estimation of the channel. In Fig.~\ref{fig:OptLenSNR}, a decreasing trend of $T_{\cab}^{*}$ with respect to the increase of $\snr$ is observed.

\begin{figure}[htbp]
\centering
\includegraphics[width=8cm, height=3.5cm]{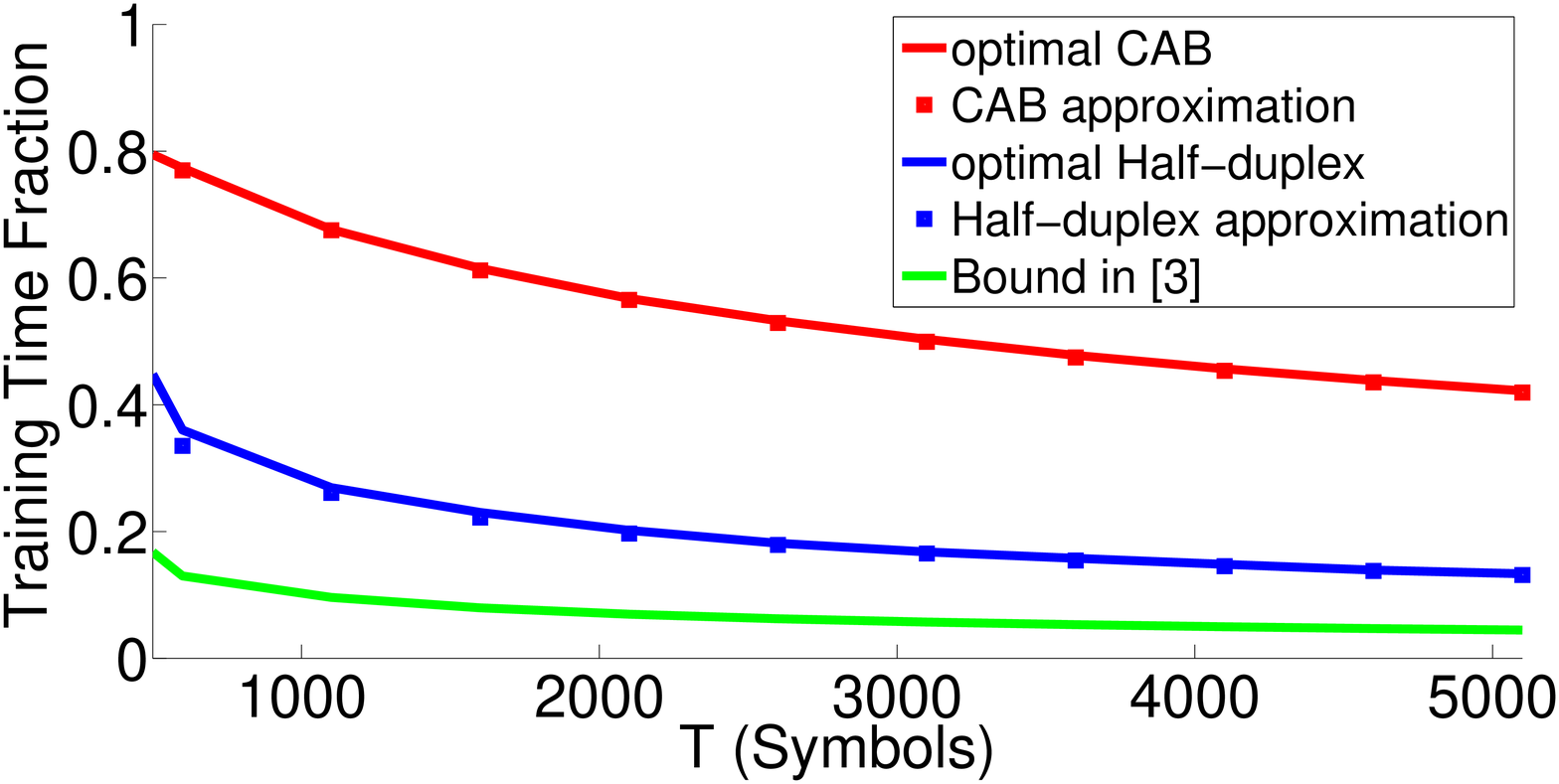}
\caption{Optimal fraction of time spent in training vs. block length, where $\snr=10\text{dB}$, $M=8$, $f=\alpha=0.1$.} \label{fig:OptLenT}
\end{figure}

\subsubsection{Half-duplex System}
We now study the optimal feedback duration for the half-duplex system, where a tighter approximation than that of \cite{kobayashi2011training} is presented.
\begin{theorem}
The optimal training duration $T_{\hd}^*$ that maximizes the spectral efficiency of the open-loop training half-duplex system is tightly approximated as
 \begin{equation}
\widetilde{T}_{\hd}^{*}=\frac{\sqrt{4cT+1}-1}{2c}\approx \sqrt{\frac{\left(M-1\right)T}{f R^{\zf}}},
\end{equation}
where $c=\frac{f R^{\zf}}{M-1}.$
\end{theorem}
\begin{proof}
The proof follows a similar approach to that of Theorem~\ref{thm:LenCAB}. The spectral efficiency benefit and loss of adding one more symbol to each cycle is first captured. The optimal training duration of half-duplex also achieves when the benefit and loss breaks even. Detailed proof can be found in Appendix~\ref{app:LenHD}.
\end{proof}

Similar to CAB system, an identical scaling trend of $T_{\hd}^*$ with respect to $T$, $\snr$ and $f$ is observed. Numerical simulation  results shown in Figs.~\ref{fig:OptLenT} and \ref{fig:OptLenSNR} confirm this observation. Detailed parameters can be found in the respective captions. It is also worthwhile noting that since $R^{\zf} \geq \widetilde{\Delta R}_{\ini}$, this suggests a longer feedback duration for CAB system. 
\begin{figure}[htbp]
\centering
\includegraphics[width=8cm, height=3.5cm]{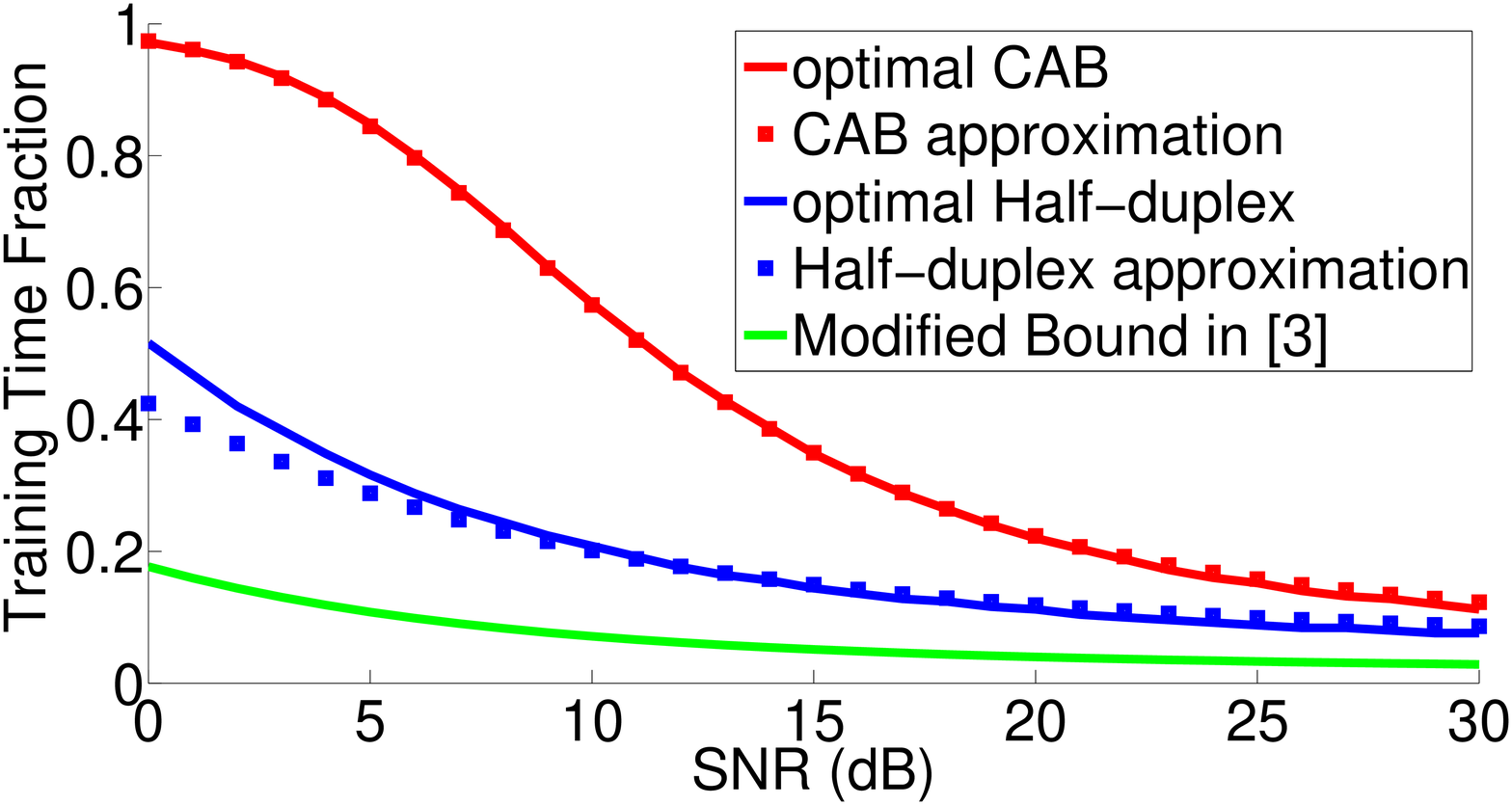}
\caption{Optimal fraction of time spent in training vs SNR. For T, where $T=2000$, $M=8$, $f=\alpha=0.1$.} \label{fig:OptLenSNR}
\end{figure}
 \subsection{Spectral Efficiency} \label{sec:SpEf}

The optimal training length analysis in Section~\ref{sec:OptLen} motivates an analytical study of spectral efficiency and potential gains reaped by CAB with open-loop training. 

\subsubsection{Open-loop CAB}
The maximal spectral efficiency of CAB strategy is achieved when training duration is optimized, i.e. $T_{\cab}^*$ symbols are used for training. We define this as \textsl{optimal CAB} strategy. It is straightforward to observe that the optimal  spectral efficiency is lower bounded by a CAB strategy with $\widetilde{T}_{\cab}^*$ symbols, i.e., $AR_{\cab}\left(T_{\cab}^*\right) \geqslant AR_{\cab}\left(\widetilde{T}_{\cab}^*\right).$ That leads to the following theorem.
\begin{theorem} \label{thm:CABspe}
The spectral efficiency loss of CAB system with respect to genie-aided full-CSI system is upper-bounded as
\begin{align}
\Delta  AR_{\cab}\left(T_{\cab}^*\right)\leqslant& R^{\zf}-AR_{\cab}\left(\widetilde{T}_{\cab}^*\right) \notag\\
\leqslant& 2\sqrt{\frac{M\widetilde{\Delta R}_{\ini}}{fT}}+o(\frac{1}{\sqrt{T}}). \label{equ:dARCAB}
\end{align} 
\end{theorem}
\begin{proof}
Evaluating the spectral efficiency of CAB with approximated optimal training duration and with futher mathematical manipulation leads to the theorem. Details of the mathematical steps can be found in Appendix~\ref{app:speCAB}.
\end{proof}

Here $o(\frac{1}{\sqrt{T}})$ suggests the spectral influence of such term vanishes in systems with large $T$. Another lower bound is the case where $T_{\cab}=T^*_{\hd}$, thus the CAB takes the same training length as optimal half-duplex system. We refer this scheme as ``\textsl{CAB with $T_{\cab}=T^*_{\hd}$}."

\subsubsection{Half-duplex System}
Similar to CAB strategy, the spectral efficiency loss (with respect to genie-aided system) of half-duplex system can also be upper bounded by evaluating $AR_{\hd}\left(\widetilde{T}_{\hd}^*\right)$.  
\begin{theorem}
The upper bound for the spectral efficiency loss of half-duplex systems with respect to genie-aided system can be characterized as
\begin{equation}
\Delta  AR_{\hd}\left(T_{\hd}^*\right)\leqslant R^{\zf}-AR_{\hd}\left(\widetilde{T}_{\hd}^*\right)\leqslant 2\sqrt{\frac{(M-1) R_{\zf}}{fT}}.
\end{equation} 
\end{theorem}
\begin{proof}
Similar to the spectral efficiency characterization of CAB systems, we can directly capture the spectral efficiency of half-duplex systems by evaluating systems with approximated training duration. Detailed derivation can be found in Appendix~\ref{app:seHD}.
\end{proof}

While larger INI widens the gap in~\eqref{equ:dARCAB}, the term $\widetilde{\Delta R}_{\ini}$ is strictly smaller than $R_{\zf}$ for any $\alpha>0$. That is, CAB strategy with open-loop training always implies a \emph{positive} spectral efficiency gain over its half-duplex counterpart under any INI regime. Numerical simulation results in Fig.~\ref{fig:ARplot} agrees with the analysis.  However, it should be acknowledged that the potential spectral benefit of CAB decrease as INI level increases. In such situations, more sophisticated techniques, such as decode-and-cancel, are needed to properly handle INI. 

It is worthwhile noting that a spectral efficiency loss scaling at the rate of $\frac{1}{\sqrt{T}}$ is observed for both full-duplex and half-duplex systems with open-loop training. This scaling law also exists in half-duplex MIMO broadcast channels with analog feedback \cite{kobayashi2011training}.  

The spectral efficiency gain of CAB can be immediately captured by comparing a CAB strategy with its half-duplex counterpart of the same training length $T_{\tr}$.
\begin{Proposition}
 The spectral efficiency gain $ARG=AR_{\cab}\left( T_{\tr}\right)-AR_{\hd}\left(T_{\tr}\right) $ is lower bounded as follows 
$$ARG \geqslant \frac{T_{\tr}}{T}\left[R_{\zf}-\widetilde{\Delta R}_{\ini}\right]-o(\frac{1}{\sqrt{T}}).$$
 \end{Proposition}
Proof details are provided in Appendix~\ref{app:seGain}. The above results manifest a positive spectral efficiency gain is always available in full-duplex systems compared to half-duplex, and the gain is significant in low INI regimes. 
  \begin{figure}[htbp]
\centering
\includegraphics[width=8cm, height=3.5cm]{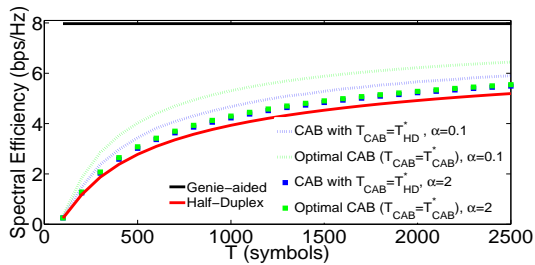}
\caption{Achievable rate performance for genie-aided, CAB and half-duplex system, where $M=8$, $\snr=10\text{dB}$, $f=0.1$.} \label{fig:ARplot}
\end{figure}
The simulation result in Fig.~\ref{fig:ARin} demonstrates these findings. Optimal allocation of the time resource for CAB strategy reveals spectral efficiency improvement over the CAB with $T_{\cab}=T_{\hd}^*$ strategy. We observe that the spectral efficiency of optimal CAB can be nearly doubled for MIMO broadcast channels around $0$ dB. We note that the spectral efficiency gain decreases with $\snr$. This is also marked through the decrease of optimal training symbols with $\snr$, analyzed in Section~\ref{sec:OptLen}.
  \begin{figure}[htbp]
\centering
\includegraphics[width=8cm,height=3.5cm]{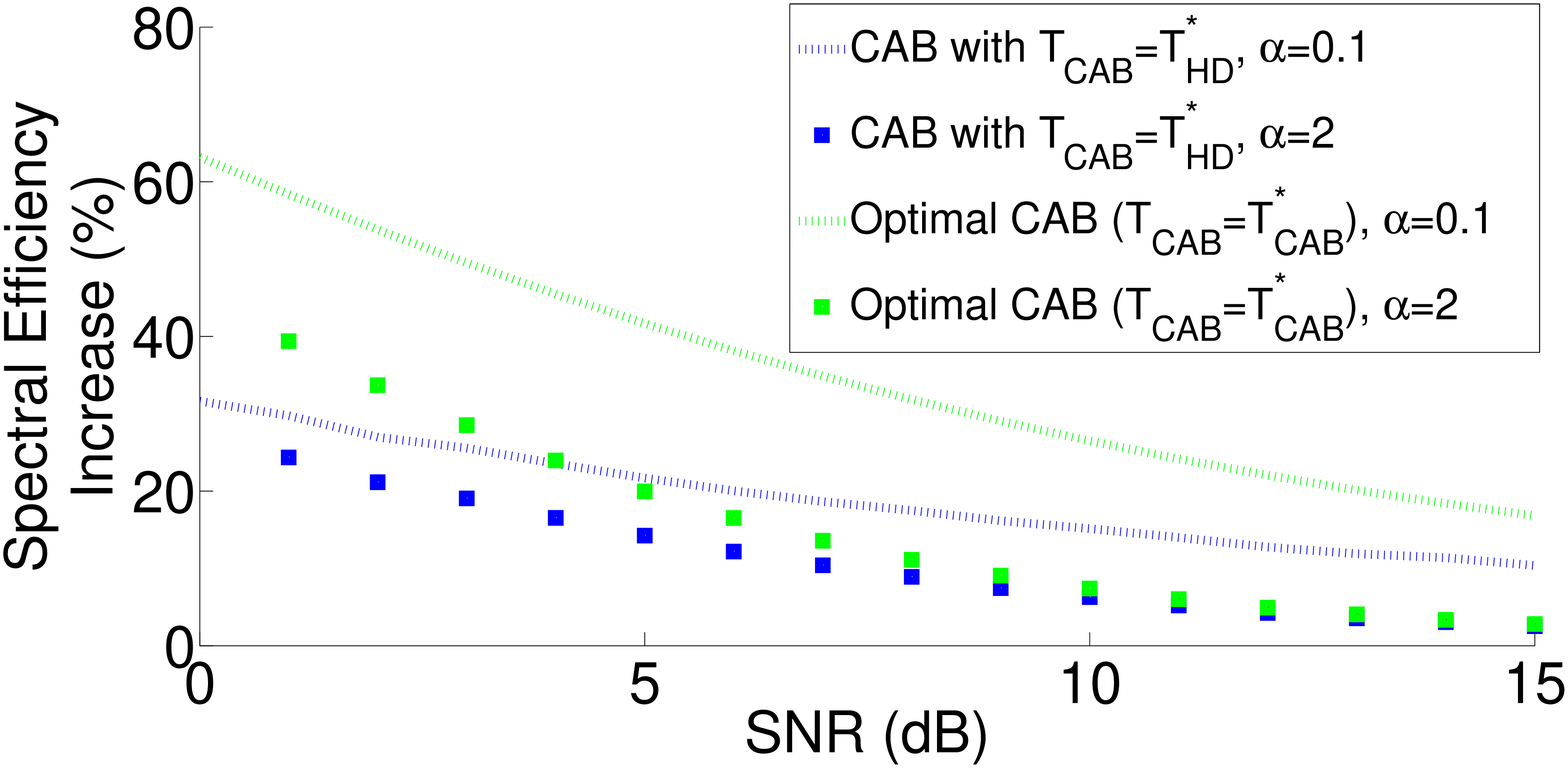}
\caption{The spectral efficiency gain compared to half-duplex counter part (in percentage) is shown for optimal CAB ans CAB with $T_{\cab}=T_{\hd}^*$ system, where $M=8$, $T=2000$, $f=0.1$.} \label{fig:ARin}
\end{figure}

\section{Conclusion} \label{sec:conclude}
We extend CAB of full-duplex communications to a multiuser MIMO downlink system with open-loop training which harnesses channel reciprocity for feedback. The effect of INI is analyzed and the optimal tradeoff between INI and IBI is investigated. We provide a novel technique to tightly approximate the optimal allocation of training symbols to both full-duplex and half-duplex systems. We further analytically characterize the spectral efficiency of CAB and its half-duplex counterpart. We demonstrate positive spectral efficiency gain of the proposed CAB strategy over half-duplex scenarios, and show that these gains are significant in low-to-moderate $\snr$ regimes whereby the base station suffers inaccurate estimation of channel CSI. Analytical results are supported with numerical simulations to demonstrate the potential gains.

\begin{appendices}
\section{INI characterization}\label{app:INI}
We will upper bound the rate loss $\Delta R_{\ini} \left(\beta M\right)$. Following the notation in \cite{caire2010multiuser}, we will use $\mathbf{v}_k$ for the precoding vector of user $k$. Since we consider perfect channel knowledge at the receiver side, then the rate loss can be upper-bounded as follows.
\begin{align*} 
\Delta R_{\ini}(\beta M) \stackrel{a)}{\leqslant} & \mathbb{E}\left[\log\left(1+\frac{|\mathbf{h}_{k}\mathbf{v}_k|^2P}{M}\right)\right]
\\&-\mathbb{E}\left[\log\left(1+\frac{|\mathbf{h}_{k}\mathbf{v}_k|^2 P/M}{1+\mathcal{P}_{\mathrm{IBI}}(\beta)+\alpha f P}\right)\right]
\\ \stackrel{b)}{\leqslant} &\mathbb{E}\left[\log\left(1+\frac{|\mathbf{h}_{k}\mathbf{v}_k|^2P}{M}\right)\right]
\\&-\mathbb{E}\left[\log\left( 1+\frac{|\mathbf{h}_{k}\mathbf{v}_k|^2P}{M} +\alpha f P\right)\right]
\\&+\mathbb{E}\left[\log\left( 1+\mathcal{P}_{\mathrm{IBI}} (\beta)+\alpha f P\right)\right]
\\ \stackrel{c)}{\leqslant}&  \log\left( \frac{1+\frac{P}{M}\frac{M-1}{1+\beta f P} +\alpha f P}{1 +\frac{\alpha f P}{1+\frac{P}{M}}}\right).  
\end{align*}
Here, a) is obtained directly from \cite{caire2010multiuser}, and  b) follows by ignoring the positive  term $\mathcal{P}_{\mathrm{IBI}}(\beta)$\footnote{$\mathcal{P}_{\mathrm{IBI}}$ is the IBI interference due to imperfect CSI of $\beta$ symbols and is explicitly detailed in \cite{caire2010multiuser}.}  in the numerator. Finally, c) holds by applying Jensen's inequality while noting the concavity of the $\log$ function.

\section{Optimal Training Length for Half-duplex system}\label{app:LenHD}
Since $AR_{\hd}$ is also concave (with respect to $T_{\hd}$), the unique optimizer $T_{\hd}^*$ can be obtained by solving
$\frac{\partial AR_{\hd}\left(T_{\hd}^*\right)}{\partial T_{\hd}}=0$.

Despite taking derivative is possible for $AR_{\hd}$, closed form expression for $T^*_{\hd}$ is unavailable. Hence, we will use a similar approach to that in Section III-C.
We first expand the expression of $AR_{\hd}$ and use $\Delta R_{\data}$ to evaluate $R_{\data}$.
\begin{align}
& \frac{T-\left(\widetilde{T}_{\hd}^*+M\right)}{T}\left(\Delta R_{\data}\left(\widetilde{T}_{\hd}^*-\Delta R_{\data}\left(\widetilde{T}_{\hd}^*+M\right)\right)\right)\notag\\=&\frac{M}{T}R_{\data}\left(\widetilde{T}_{\hd}^*+M\right), \label{equ:optHDcom}
\end{align} 
On the left hand side is the same as Eq. (8), which is the marginal utility for adding one more training cycle. On the right hand is the potential spectral efficiency contribution (marginal utility) of this $M$ symbols. A higher marginally loss is observed. This agrees with the fact that time resources used for training (in half-duplex system) has no downlink throughput.

Noticing that the left hand side is identical to the LHS term in~\eqref{equ:optHDcom}, we can directly borrow the evaluation of $\text{LHS}$. $R_{\data}\left(T_{\hd}+M\right)$ can be naively approximated by $R^{\zf}$. These two results complete the proof.

\section{Upper bound of CAB strategy spectral efficiency loss}\label{app:speCAB}

\begin{align*}
\Delta  AR_{\cab}\left(T_{\cab}^*\right)\leqslant R^{\zf}-AR_{\cab}\left(\widetilde{T}_{\cab}^{*}\right).
\end{align*}
We will expand the $AR_{\cab}$ with the help of rate characterization in the main paper. Since the downlink rate during each cycle is non-negative, we can also drop the spectral efficiency gain in the second cycle. With some mathematical manipulations gives us the following form. 
\begin{align*} 
&\Delta  AR_{\cab}\left(T_{\cab}^*\right)\\
\leqslant& \frac{2M}{T}R^{\zf}+\sum_{\beta=2}^{\frac{\widetilde{T}_{\cab}^{*}}{M}-1}\frac{M}{T}\log\left( \frac{1+\frac{P}{ M}\frac{M-1}{1+\beta f P } +\alpha f P}{1 +\frac{\alpha f P}{1+\frac{P}{M}}}\right)  + \left(1-\frac{\widetilde{T}_{\cab}^{*}}{T}\right)\Delta R_{\data}\left(\widetilde{T}_{\cab}^{*}\right).
\end{align*}

We will then analyses each term independently. Since $\lim_{T\to \infty}\frac{2M}{T}R^{\zf}\sqrt{T}=0$, it is immediate to denote it as $o(\frac{1}{\sqrt{T}})$. The third term can then be upper bounded as follows.

\begin{align*}
\left(1-\frac{\widetilde{T}_{\cab}^{*}}{T}\right)\Delta R_{\data}\left(\widetilde{T}_{\cab}^{*}\right)
\leqslant& \Delta R_{\data}\left(\widetilde{T}_{\cab}^{*}\right)\\
=&\log\left(  1+\frac{\frac{M-1}{M}P}{1+ f  \widetilde{T}_{\cab}^{*}P/M}\right) \\
\leqslant& \frac{ M-1}{ f  \widetilde{T}_{\cab}^{*} }.
\end{align*}
The final step is done by using the Maclaurin expansion.

The second term can be then upper bounded by removing $1$ in the denominator and change $M-1$ into $M$ in the numerator, which leads to

\begin{align*}
\sum_{\beta=2}^{\frac{\widetilde{T}_{\cab}^{*}}{M}-1}\frac{M}{T}\log\left( \frac{1+\frac{P}{M}\frac{M-1}{1+\beta f P} +\alpha f P}{1 +\frac{\alpha f P}{1+\frac{P}{M}}}\right) \leqslant\sum_{\beta=2}^{\frac{\widetilde{T}_{\cab}^{*}}{M}-1}\frac{M}{T}\log\left( \frac{1+\frac{M-1}{M}\frac{1}{ \beta f}+\alpha f P}{1 +\frac{\alpha f P}{1+\frac{P}{M}}}\right)
\end{align*}

Since $\log\left(1+x+y\right)\leqslant\log(1+x)+\log(1+y)$ for $x,y>0$, this implies we can break the rate loss term into 2 parts. 
\begin{equation*}
\sum_{\beta=2}^{\frac{\widetilde{T}_{\cab}^{*}}{M}-1}\frac{M}{T}\log\left( \frac{1+\frac{M-1}{M}\frac{1}{ \beta f}+\alpha f P}{1 +\frac{\alpha f P}{1+\frac{P}{M}}}\right) \leqslant \sum_{\beta=2}^{\frac{\widetilde{T}_{\cab}^{*}}{M}-1}\frac{M}{T}\left[\log\left( \frac{1+\alpha f P/M}{1 +\frac{\alpha f P}{1+\frac{P}{M}}}\right)+\log\left( 1+\frac{M-1}{M}\frac{1}{ \beta f }\right)\right].
\end{equation*}
Noticing that $\log\left( \frac{1+\alpha f P/M}{1 +\frac{\alpha f P}{1+\frac{P}{M}}}\right)=\widetilde{\Delta R}_{\ini}$ and enlarging $M-1$ in the numerator to $M$ will simplify the bound to
$$\frac{\widetilde{T}_{\cab}-2M}{T} \widetilde{\Delta R}_{\ini}+\sum_{\beta=2}^{\frac{\widetilde{T}_{\cab}^{*}}{M}-1}\log\left( 1+\frac{M-1}{M}\frac{1}{ \beta f }\right) $$.

Due to the convexity of $\log(1+x)$, we can apply Jensen's inequality to the latter term. The two above steps gives us
\begin{align*}
\sum_{\beta=2}^{\frac{\widetilde{T}_{\cab}^{*}}{M}-1}\frac{M}{T}\log\left( \frac{1+\frac{P}{M}\frac{M-1}{1+\beta f P} +\alpha f P}{1 +\frac{\alpha f P}{1+\frac{P}{M}}}\right)\leqslant\frac{\widetilde{T}_{\cab}-2M}{T}\left[\widetilde{\Delta R}_{\ini}+ \log\left( 1+\frac{\sum_{\beta=2}^{\frac{\widetilde{T}_{\cab}^{*}}{M}-1}\frac{M-1}{M}\frac{1}{ \beta} }{f\left(\widetilde{T}_{\cab}^{*}-2\right)} \right)\right].
\end{align*} 

Enlarging $\frac{\widetilde{T}_{\cab}-2M}{T}$ to $\frac{\widetilde{T}_{\cab}}{T}$ and expanding the later term gives us
\begin{align*}
\sum_{\beta=2}^{\frac{\widetilde{T}_{\cab}^{*}}{M}-1}\frac{M}{T}\log\left( \frac{1+\frac{P}{M}\frac{M-1}{1+\beta f P} +\alpha f P }{1 +\frac{\alpha f P }{1+\frac{P}{M}}}\right)\leqslant \frac{\widetilde{T}_{\cab}}{T}\widetilde{\Delta R}_{\ini}+\frac{\widetilde{T}_{\cab}}{T}  \frac{M-1}{M}\frac{1}{f}\frac{ \log\left(\widetilde{T}_{\cab}^{*}-1\right)}{\left(\widetilde{T}_{\cab}^{*}-2\right)}\\
\leqslant  \sqrt{\frac{M\widetilde{\Delta R}_{\ini}}{fT}}+o(\frac{1}{\sqrt{T}}).
\end{align*}

Substituting the upper bound we got for the three individual term, a total upper bound for the spectral efficiency loss can be expressed as

\begin{align*}
\Delta AR_{\cab}\left(T_{\cab}^*\right)\leqslant& o(\frac{1}{\sqrt{T}})+  \sqrt{\frac{M\widetilde{\Delta R}_{\ini}}{fT}}+\sqrt{\frac{M\widetilde{\Delta R}_{\ini}}{fT}}+o(\frac{1}{\sqrt{T}})\\
=& 2\sqrt{\frac{M\widetilde{\Delta R}_{\ini}}{fT}}+o(\frac{1}{\sqrt{T}}).
\end{align*} 

\section{Upper bound of Half-duplex system spectral efficiency loss}\label{app:seHD}
 The gap with respect to $R^{\zf}$ can be immediately obtained by following the recipe in~\cite{kobayashi2011training} Section II-A as: 
 
 \begin{align*}
\Delta AR_{\hd}\left(T_{\hd}^*\right)\leqslant& R^{\zf}-AR_{\hd}\left(\widetilde{T}_{\hd}^{*}\right)\\
 =& R^{\zf}-\left(\frac{T-T_{\hd}^*}{T}\right)\left[R^{\zf}-\log\left(  1+\frac{(M-1)\frac{P}{M}}{1+ f  \widetilde{T}_{\hd}^{*}P/M}\right)\right]\\
\stackrel{a)}{\leqslant}& \frac{\widetilde{T}_{\hd}^{*}}{T}R^{\zf}+\log\left(  1+\frac{(M-1)\frac{P}{M}}{1+ f  \widetilde{T}_{\hd}^{*}P/M}\right) \\
\leqslant& \sqrt{\frac{(M-1)R^{\zf}}{fT}}+\sqrt{\frac{R^{\zf}(M-1)}{fT}}\\
=& 2\sqrt{\frac{(M-1)R^{\zf}}{fT}}.
\end{align*}
Step a) is obtained by drooping the negative term $-\frac{\widetilde{T}_{\hd}^{*}}{T}\log\left(  1+\frac{(M-1)\frac{P}{M}}{1+ f  \widetilde{T}_{\hd}^{*}P/M}\right)$. The next step is the result of
Maclaurin expansion to the logarithm term, which is tight for big $T$.

\section{Lower Bound on spectral efficiency of sub-optimal CAB strategy}\label{app:seGain}
  
 \begin{align*}
 & AR_{\cab}\left( T_{\hd}^*\right)-AR_{\hd}\left(T_{\hd}^*\right)\\
 \geqslant& \sum_{\beta=2}^{\frac{T_{\hd}^{*}}{M}-1}\frac{M}{T}\left[ R^{\zf}-\Delta R_{\ini}\left(\beta\right)\right]  \\
 \geqslant&\frac{T_{\hd}^{*}-2M}{T}\left[R_{\zf}-\widetilde{\Delta R}_{\ini}\right]- \frac{ T_{\hd}^*}{T}  \frac{M-1}{M}\frac{1}{f}\frac{ \log\left( T_{\hd}^*-1\right)}{\left( T_{\hd}^*-2\right)}\\
 \geqslant& \frac{T_{\hd}^{*}}{T}\left[R_{\zf}-\widetilde{\Delta R}_{\ini}\right]-o(\frac{1}{\sqrt{T}}).
 \end{align*} 
 The step is directly obtained by dropping the data transmission during the second cycle. The next step follow the same characterization method used in the proof of Theorem~\ref{thm:CABspe}. The finally step is the result of combining two $o(\frac{1}{\sqrt{T}})$ terms. 
 
\end{appendices}

 \bibliography{IEEEabrv,SPAWC15Ref}

\end{document}